\documentclass[12pt]{article}
\usepackage[ruled, boxed]{algorithm}
\usepackage{algpseudocode}
\usepackage[flushleft]{threeparttable}

\usepackage{amsmath, amsthm, amssymb, amstext, comment, graphicx, fullpage}
\usepackage{xspace}

\usepackage{amsmath,amssymb}
\usepackage{verbatim}
\usepackage{enumerate}

\newtheorem{definition}{Definition}

\newtheorem{example}{Example}

\newtheorem{theorem}{Theorem}
\renewcommand{\vec}{\mathbf}

\DeclareMathOperator*{\argmin}{\arg\!\min}

\DeclareMathOperator*{\argmax}{arg\,max}
\usepackage{pifont}% http://ctan.org/pkg/pifont

\title{Characterization and Computation of Equilibria for Indivisible Goods}

\author{
\textbf{Simina Br\^anzei}\\
\small{Aarhus University, Denmark}\\
\small{\texttt{simina@cs.au.dk}}\\
\newline
\and
\textbf{Hadi Hosseini}\\
\small{University of Waterloo, Canada}\\
\small{\texttt{h5hossei@uwaterloo.ca}}
\and
\textbf{Peter Bro Miltersen}\\
\small{Aarhus University, Denmark}\\
\small{\texttt{bromille@cs.au.dk}}
}
\date{}
\begin{document}
\maketitle

\begin{abstract}
We consider the problem of allocating indivisible goods in a way that is fair, using one of the leading market mechanisms in economics: \emph{the competitive equilibrium
from equal incomes}, which can be interpreted as a notion of fairness. Focusing on two major classes of valuations, namely \emph{perfect substitutes} and \emph{perfect complements}, we establish the computational properties of algorithms operating in this framework. 
For the class of valuations with perfect complements, our algorithm yields a surprisingly succinct characterization of instances that admit a competitive equilibrium from equal incomes.
\end{abstract}

\thispagestyle{empty}
\setcounter{page}{0}

\newpage

\section{Introduction}
The systematic study of economic mechanisms began in the 19th century with the pioneering work of Irving Fisher~\cite{BS00} and L\'{e}on Walras \cite{Walras74},
who proposed 
the Fisher market and the exchange economy as answers to the question: ``How does one allocate scarce resources among the participants of an economic system?''. These models of a competitive economy are central in mathematical economics and have been studied ever since in an extensive body of literature~\cite{AGT_book}.

The high level scenario is that of several economic players arriving at the market with an initial endowment of resources and a utility function for consuming goods.
The problem is to compute prices and an allocation for which an optimal exchange takes place: each player is maximally satisfied with the bundle acquired, given the prices and his initial endowment.
Such allocation and prices form a \emph{market equilibrium} and, remarkably, are guaranteed to exist under mild assumptions when goods are divisible~\cite{AD54}.

In real scenarios, however, goods often come in discrete quantities; for example, clothes, furniture, houses, or cars may exist in multiple copies, but cannot be infinitely divided. Scarce resources, such as antique items or art collection pieces are even rarer -- often unique (and thus \emph{indivisible}). 
The problem of allocating discrete or indivisible resources is much more challenging because the theoretical guarantees from the divisible case do not always carry over;
however, it 
can be tackled as well using market mechanisms~\cite{DPS03,B11,BL14,OPR14}. In this paper, we are concerned with the question of allocating indivisible resources
using the \emph{competitive equilibrium from equal incomes} (CEEI).

The competitive equilibrium from equal incomes solution embodies the ideal notion of fairness ~\cite{Fol67,Var74,hylland1979efficient,OPR14} and is a special case of the Fisher market model~\cite{TH85}. Informally, there are $m$ goods to be allocated 
among $n$ buyers, each of which is endowed with one unit of an artificial currency that they can use 
to acquire goods. The buyers declare their preferences over the goods, after which the equilibrium prices and allocation are computed. 
When the goods are divisible, a competitive equilibrium from equal incomes is guaranteed to exist for very general conditions and each equilibrium allocation satisfies the desirable properties of envy-freeness and efficiency.

In recent years, the competitive equilibrium from equal incomes has also been studied as a mechanism for the fair allocation of discrete and indivisible resources in a series of papers.
Bouveret and Lema\^{\i}tre \cite{BL14} considered it for allocating indivisible goods, together with several other notions of fairness such as proportionality, envy-freeness, and maximin fairness. 
Budish \cite{B11} analyzed the allocation of multiple discrete goods for the course assignment problem\footnote{Given a set of students and courses to be offered at a university,
how should the courses be scheduled given that the students have preferences over their schedules and the courses have capacity constraints on enrollment?}
and designed an approximate variant of CEEI that is guaranteed to exist for any instance.
In this variant, buyers have permissible bundles of goods and the approximation notion
requires randomization to perturb the budgets of the buyers while relaxing the market clearing condition. In follow-up work, Othman, Papadimitriou, and Rubinstein \cite{OPR14} analyzed the computational complexity of
this variant,
showing that computing the approximate solution proposed by Budish is PPAD-complete, and that it is NP-hard to distinguish between an instance where an exact CEEI exists and the one in which there is no approximate-CEEI tighter than guaranteed in Budish \cite{B11}.

In this paper, we study the competitive equilibrium from equal incomes 
for two major classes of valuations, namely \emph{perfect substitutes} and \emph{perfect complements}.
Perfect substitutes represent goods 
that can replace each other in consumption, such as Pepsi and Coca-Cola, and are modeled mathematically 
through \emph{additive utilities}. This is the setting examined by Bouveret and Lema\^{\i}tre \cite{BL14} as well.
Perfect complements represent goods 
that have to be consumed together, such as a left shoe and a right shoe, and are modeled mathematically through \emph{Leontief utilities}. For indivisible goods, Leontief utilities are in fact equivalent to the class of single-minded buyers, which have been studied extensively in the context of auctions~\cite{AGT_book}. 

We study the computation of competitive equilibria for indivisible goods and establish polynomial time algorithms and hardness results (where applicable). Our algorithm for Leontief utilities gives a very 
succinct characterization of markets that admit a competitive equilibrium from equal incomes for indivisible resources. For additive valuations we obtain a series of hardness results.
The computational results of Othman, Papadimitriou, and Rubinstein \cite{OPR14} are orthogonal to our setting since they refer to combinatorial valuations.

\section{Competitive Equilibrium from Equal Incomes}

We begin by formally introducing the competitive equilibrium from equal incomes. Formally, there is a set $N = \{1, \ldots, n\}$ of buyers and a set $M = \{1, \ldots, m\}$ of goods which are brought by a seller. 
In general, the goods can be either infinitely divisible or discrete and, without loss of generality, there is exactly one unit from every good $j \in M$.
Each buyer $i$ is endowed with:
\begin{itemize}
\item A utility function $u_i : [0,1]^m \rightarrow \mathbb{R}_{\geq 0}$ for consuming the goods, which maps each vector $\vec{x} = \langle x_1, \ldots, x_m \rangle$ of resources to a real value, where $u_i(\vec{x})$ represents the buyer's utility for bundle $\vec{x}$; note that $\vec{x}_{j}$ is the amount received by the buyer from good $j$.
\item An initial budget $B_i = 1$, which can be viewed as (artificial) currency to acquire goods, but has no intrinsic value to the buyer. However, the currency does have intrinsic value to the seller.
\end{itemize}

Each buyer in the market wants to spend its entire budget to acquire a bundle of items that maximizes its utility, while the seller aims to sell all the goods (which it has no intrinsic value for) and extract the money from the buyers. 

A market outcome is defined as a tuple $(\vec{x},\vec{p})$, where 
$\vec{p}$ is a vector of prices for the $m$ items, and 
$\vec{x} = \langle \vec{x}_1, \ldots, \vec{x}_n \rangle$ is
an allocation of the $m$ items, with $p_j$ denoting the price of item $j$ and $x_{ij}$
representing the amount of item $j$ received by buyer $i$. A market outcome that 
maximizes the utility of each buyer subject to its budget constraint
and clears the market is called a \emph{market equilibrium}~\cite{AGT_book}.
Formally, $( \vec{x},\vec{p} )$ is a market equilibrium if and only if:
\begin{itemize}
\item For each buyer $i \in N$, the bundle $\vec{x}_i$ maximizes buyer $i$'s utility given the prices $\vec{p}$ and budget $B_i = 1$.

\item Each item $j \in M$ is completely sold or has price zero. That is: $$\left(\sum_{i=1}^{n} x_{ij}-1\right)p_j =0.$$
\item All the buyers exhaust their budgets; that is, $\sum_{j=1}^{m} p_j \cdot x_{ij}=1$, for all $ i \in N$.
\end{itemize}

Every competitive equilibrium from equal incomes $(\vec{x}, \vec{p})$ is envy-free; if buyer $i$ would strictly prefer 
another buyer $j$'s bundle $\vec{x}_j$, then $i$ could simply purchase $\vec{x}_j$ instead of $\vec{x}_i$ since they have the same buying power, which is in contradiction with the equilibrium property.

A market with divisible goods is guaranteed to have a competitive equilibrium under mild conditions~\cite{AD54}.
Moreover, for the general class of \emph{Constant Elasticity of Substitution} valuations, the equilibrium can be computed using a remarkable convex program due to Eisenberg and Gale \cite{EG59}, as one of the few algorithmic results in general equilibrium theory.
The classes of valuations studied in this paper -- perfect complements and perfect substitutes -- belong to the constant elasticity of substitution family.\\

In the following sections we study these classes in detail in the context of allocating indivisible resources. Our findings are summarized in Table~\ref{tab:results}. \\

%\begin{figure}[htp]
%\centering
%\includegraphics[scale=0.75]{market}
%\caption{Summary : 
%The market instance is denoted by a tuple $\mathcal{M} = \langle N, M, \vec{v} \rangle$, where $N$ is a set of buyers, $M$ a set of indivisible items, and $\vec{v}$ the valuations of the buyers for the items; 
%$\vec{x}$ is an allocation of the items to the buyers and $\vec{p}$  a price vector. The entries denote the complexity of deciding if there exists a competitive equilibrium at the given input.}
%\label{fig:results}
%\end{figure}
\begin{center}
%\begin{comment}
%	\begin{table}[h!]
\begin{threeparttable}
		%\scriptsize
		\centering
		\begin{tabular}{||l || c | c ||}
			%\begin{minipage}
			\hline \hline
			\textbf{Input} $\backslash$ \textbf{Valuations}& \begin{tabular}{c} \textbf{Perfect}\\ \textbf{Complements} \end{tabular} & \begin{tabular}{c}\textbf{Perfect}\\ \textbf{Substitutes} \end{tabular}\\ \hline \hline
			Market $\mathcal{M}$ & $\mathcal{P}$ \tnote{a} & $\mathcal{NP}$-hard \\ 
			\hline
			Market $\mathcal{M}$, allocation $\vec{x}$ & $\mathcal{P}$ & co-$\mathcal{NP}$-hard \\ 
			\hline
			Market $\mathcal{M}$, prices $\vec{p}$ & co-$\mathcal{NP}$-complete & co-$\mathcal{NP}$-hard \\
			\hline
			Market $\mathcal{M}$, allocation $\vec{x}$, prices $\vec{p}$ & $\mathcal{P}$ & co-$\mathcal{NP}$-complete \\ 
			\hline \hline
		\end{tabular}
		\begin{tablenotes}
		\item[a] For this entry we obtain a succinct characterization of instances that admit an exact CEEI, as well as a $1/n$ approximation of the optimal social welfare (without normalization).
		\end{tablenotes}
		%\end{minipage}
		\caption{Summary of the computational results. The market instance is denoted by a tuple $\mathcal{M} = \langle N, M, \vec{v} \rangle$, where $N$ is a set of buyers, $M$ a set of indivisible items, and $\vec{v}$ the values of the buyers for the items; $\vec{x}$ is an allocation of the items to the buyers and $\vec{p}$ a price vector.}
		\label{tab:results}
%	\end{table}
\end{threeparttable}
\end{center}
	%\end{comment}

\section{Perfect Complements}
Let 
$\mathcal{M} = (N, M, \vec{v})$ denote a market with perfect complements, represented through Leontief utilities; recall $N$ is the set of buyers, $M$ the set of items, and $\vec{v}$ a matrix
of constants, such that $v_{i,j}$ is the value of buyer $i$ for consuming one unit of good $j$.
The utility of buyer $i$ for bundle $\vec{x} = \langle x_1, \ldots, x_m \rangle 
\in [0,1]^m$ is given by:
\begin{equation} \label{eq:LeontiefUtility}
u_i(\vec{x}) = \min_{j=1}^{m} \left( \frac{x_{j}}{v_{i,j}} \right)
\end{equation}
%where $\vec{v} = (v_{i,j})_{i\in N, j \in M}$ is a matrix of constants, such that $v_{i,j}$ represents the value of buyer $i$ for consuming one unit of good $j$.

In our model the goods are indivisible, and so $x_{i,j} \in \{0, 1\}$, for all $i,j$.
By examining Equation~\ref{eq:LeontiefUtility}, it can be observed that buyer $i$'s utility for a bundle depends solely on whether the buyer gets all the items 
that it values positively (or not). To capture this we define the notion of \emph{demand set}.

\begin{definition}[Demand Set] \label{def:demandset}
Given a CEEI market with indivisible goods and Leontief utilities, let the \emph{demand set} of buyer $i$ be the set of items that $i$ has a strictly positive value for; that is,
$D_i = \{j \in M \; | \; v_{i,j} > 0\}$.
\end{definition}

Now we can introduce the precise utility equation for indivisible goods with Leontief valuations.
\begin{definition}[Leontief Utility for Indivisible Goods]
Given a market with Leontief utilities and indivisible goods, the utility of a buyer $i$ 
for a bundle $\vec{x} = \langle x_1, \ldots, x_m \rangle \in \{0,1\}^m$ is:
\[
u_i(\vec{x}) = 
\left\{
                \begin{array}{ll}
                  \min_{j \in D_i} \left( \frac{1}{v_{i,j}} \right), \; \mbox{if} \; D_i \subseteq \vec{x}, \\
                  0, \; \mbox{otherwise} \\
                \end{array}
              \right.
\]
where $D_i$ represents buyer $i$'s demand set.
\end{definition}

We illustrate this utility class with an example. Note that valuations are not necessarily normalized.
\begin{example}
Let $\mathcal{M}$ be a market with buyers $N = \{1,2,3\}$, items $M = \{1,2,3,4\}$, and values: $v_{1,1}=1$, $v_{2,2} = 2$, $v_{2,4} = 3$, $v_{3,1} = 0.5$, $v_{3,2} = 2.5$, $v_{3,3} = 5$,
and $v_{i,j} = 0$, for all other $i,j$.
Recall the demand set of each buyer consists of the items it values strictly positively, and so: $D_1 = \{1\}$, $D_2 = \{2,4\}$, $D_3 = \{1,2,3\}$.
Then the utility of buyer $1$ for a bundle $S \subseteq M$ is: $u_1(S) = 0$ if $D_1 \not \subseteq S$, and $u_1(S) = \min_{j \in D_1} \left(\frac{1}{v_{1,j}}\right) = \frac{1}{v_{1,1}} = 1$ otherwise.
Similarly, $u_2(S) = 0$ if $D_2 \not \subseteq M$ and $u_2(S) = \min\left(\frac{1}{v_{2,2}}, \frac{1}{v_{2,4}}\right) = \frac{1}{3}$ otherwise.
\end{example}

Next, we examine the computation of allocations that are fair according to the CEEI solution concept. The main computational problems that we consider are : \emph{Given a market, determine whether a competitive equilibrium exists and compute it when possible}. Depending on the scenario at hand, an allocation of the resources to the buyers may have already been made (or the seller may have already set prices for the items). The questions then are to determine whether an equilibrium exists at those prices or allocations. Our algorithm for computing a competitive equilibrium for Leontief utilities with indivisible goods yields a characterization of when a market equilibrium is guaranteed to exist.

\begin{theorem} \label{lem:polyLeon}
Given a market $\mathcal{M} = (N, M, \vec{v})$
with Leontief utilities, indivisible goods, and a tuple $(\vec{x}, \vec{p})$, where $\vec{x}$ is an allocation and $\vec{p}$ a price vector, it can be decided in polynomial time if $(\vec{x}, \vec{p})$ is a market equilibrium for $\mathcal{M}$.
\end{theorem}
\begin{proof}
It is sufficient to verify that these conditions hold:
\begin{itemize}
	\item Each buyer $i$ exhausts their budget: $\sum_{j \in \vec{x}_i} p_j = 1$.

	\item Each item is either allocated or has a price of zero. %%%$\sum_{i=1}^{n} x_{i,j} = 1$, for each item $j \in M$.

	\item No buyer $i$ can afford a better bundle; that is, if $u_i(\vec{x}, \vec{p}) = 0$, then $\sum_{j \in D_i} p_j > 1$.
\end{itemize}
Clearly all three conditions can be verified in polynomial time, namely $O(mn)$, which concludes the proof.
%\simina{State the exact runtime}
%\hadi{I don't think we can get any runtime better than $O(nm)$. For the first two, we'll have to traverse the matrix, and the last one is essentially better than $nm$.}
\end{proof}

\begin{theorem} \label{thm:Leontief_Ix}
Given a market $\mathcal{M} = (N, M, \vec{v})$ with Leontief utilities, indivisible goods, and a price vector $\vec{p}$, it is co-NP-complete to decide if there exists an allocation $\vec{x}$ such that $(\vec{x}, \vec{p})$ 
is a market equilibrium for $\mathcal{M}$.
\end{theorem}
\begin{proof}
From Theorem~\ref{lem:polyLeon}, given an allocation $\vec{x}$ for $(\mathcal{M}, \vec{p})$, it can be verified in polynomial time if there exists a market equilibrium for $\mathcal{M}$ at $(\vec{x}, \vec{p})$.
To show hardness, we use the NP-complete problem $\textsc{PARTITION}$:

\begin{quote}
\emph{Given a set of positive integers $S = \{s_1, \ldots, s_m\}$, are there subsets $A, B \subset S$ such that $A \cup B = S$, $A \cap B = \emptyset$, and $\sum_{a \in A} a = \sum_{b \in B} b$? }
\end{quote}
%We show that an algorithm that decides the existence of a market equilibrium at a given price vector can also solve $\textsc{PARTITION}$. 
Given partition input $S$, construct a market $\mathcal{M} = (N, M, \vec{v})$ and price vector $\vec{p}$ as follows:
\begin{itemize}
\item Set $N = \{1, 2, 3\}$ of buyers.
\item Set $M = \{0, 1, \ldots, m\}$ of items.
\item Price vector $\vec{p} = (p_0, \ldots, p_m)$, such that $p_0 = 1$ and $p_j = \frac{2 \cdot s_j}{\sum_{l=1}^{m} s_l}$, for all $j \in \{1, \ldots, m\}$.
\item Demand sets: $D_1 = \{0\}$ and $D_2 = D_3 = \{1, \ldots, m\}$. Clearly these demand sets can be expressed through Leontief valuations -- for example, let $v_{1,0} = 1$ and $v_{1,j} = 0$, for all $j \in \{1, \ldots, m\}$, and 
$v_{2,k} = v_{3,k} = \frac{1}{m+1}$, for all $k \in \{0, \ldots, m\}$.
\end{itemize}
Note that the total price of the items in $\{1, \ldots, m\}$ is: $$p(\{1, \ldots, m\}) = \sum_{j=1}^{m} p_j = \sum_{j=1}^{m} \frac{2 \cdot s_j}{\sum_{l=1}^{m} s_l} = 2$$.

$(\implies)$ If there is a partition $(A, B)$ of $\mathcal{S}$, then we show that the allocation $\vec{x}$ given by $\vec{x}_1 = \{0\}$, $\vec{x}_2 = A$, $\vec{x}_3 = B$ is a market equilibrium:
\begin{itemize}
\item All the items are sold, since $\vec{x}_1 \cup \vec{x}_2 \cup \vec{x}_3 = \{0, 1, \ldots, m\}$, and the bundles are disjoint, since $\vec{x}_i \cap \vec{x}_j = \emptyset$, for all $i,j \in N$, $i \neq j$.
\item Each buyer gets an optimal bundle at prices $\vec{p}$, since buyer $1$ gets the best possible bundle: $u_1(\vec{x}_1, \vec{p}) = u_1(D_1, \vec{p}) = 1$, 
while buyers $2$ and $3$ cannot afford anything better, as the price of their demanded bundle is higher than their budget: $\vec{p}(D_2) = \vec{p}(D_3) = \sum_{j = 0}^{m} p_j = 3 > 1$.
\item The amount of money spent by each buyer is equal to their budget: $\vec{p}(\vec{x}_1) = 1$, $\vec{p}(\vec{x}_2) = \vec{p}(A) = 1$, and $\vec{p}(\vec{x}_3) = \vec{p}(B) = 1$.
\end{itemize}

$(\impliedby)$ On the other hand, if there is an allocation $\vec{x}$ such that $(\vec{x}, \vec{p})$ is a market equilibrium for the market $\mathcal{M}$, we claim that $A = \vec{x}_2$ and $B = \vec{x}_3$ is a correct partition of $\mathcal{S}$.
First, note that if $\vec{x}$ is such that buyer $1$ does not get item $0$, then the optimality condition fails for buyer $1$ because $p_0 = 1$ and so the buyer could always afford it.
Thus buyer $1$ must get item
$0$ and moreover, it cannot get anything else. Thus $\vec{x}_1 = \{0\}$ and $\vec{x}_2, \vec{x}_3 \subseteq \{1, \ldots, m\}$.
\begin{itemize}
\item Since all the items are sold at $(\vec{x},\vec{p})$, we have that $A \cup B = S$.
\item Buyers $2$ and $3$ spend their entire budgets at $(\vec{x}, \vec{p})$, and so $\vec{p}(\vec{x}_2) = 1$ and $\vec{p}(\vec{x}_3) = 1$. Then $\vec{p}(A) = \vec{p}(B) = 1$ and:
\begin{eqnarray*}
\sum_{s_j \in A} \left( \frac{2 \cdot s_j}{\sum_{k=1}^{m} s_k}\right) & = & \sum_{s_j \in B} \left( \frac{2 \cdot s_j}{\sum_{k=1}^{m} s_k} \right) \iff \\
\sum_{s_j \in A} s_j & = & \sum_{s_j \in B} s_j
\end{eqnarray*}
\end{itemize}
Thus $S$ has a partition if and only if the corresponding market and price vector admit a market clearing allocation, which completes the proof.
\end{proof}

%The proof of Theorem~\ref{thm:Leontief_Ix} uses the NP-complete problem $\textsc{PARTITION}$ and is relegated to the full version of the paper.
%The theorem can be proved using a reduction from the NP-complete problem $\textsc{PARTITION}$; the proof is delegated to the full version of the paper.

\begin{theorem}
Given a market $\mathcal{M} = (N, M, \vec{v})$ with Leontief utilities and an allocation $\vec{x}$, it can be decided in polynomial time if there exists price vector $\vec{p}$ such that 
$(\vec{x}, \vec{p})$ is a market equilibrium for $\mathcal{M}$.
\end{theorem}
\begin{proof}
This problem can be solved using linear programming (see Algorithm~\ref{alg1}).
At a high level, one needs to check that the allocation $\vec{x}$ is feasible, that each item is either sold or has a price of zero, and that (i) each buyer spends all their money and (ii)
whenever a buyer does not get their demand set, the bundle is too expensive.
Since the number of constraints is polynomial in the number of buyers and items, the algorithm runs in polynomial time.
\end{proof}

%\vspace{-4mm}
\begin{algorithm}[h!]
\begin{algorithmic}
\caption{\textsc{Compute-Equilibrium-Prices($\mathcal{M}, \vec{x}$)}}
\label{alg1}
\State \textbf{input}: Market $\mathcal{M}$ with Leontief valuations; allocation $\vec{x}$
\State \textbf{output}: price vector $\vec{p}$ such that $(\vec{x}, \vec{p})$ is a market equilibrium for $\mathcal{M}$, or \emph{\textsc{Null}} if none exists
\State $\mathcal{A} \leftarrow \emptyset$ \emph{// Set of items allocated under $\vec{x}$} 
\State \emph{// Check that $\vec{x}$ is feasible}
\For{$i = 1 \; \mbox{\emph{\textbf{to}}} \; n$} 
  \State $\mathcal{A} \leftarrow \mathcal{A} \cup \vec{x}_i$ 
  \For{$j = i+1 \; \mbox{\emph{\textbf{to}}} \; n$} 
    \If{($\vec{x}_i \cap \vec{x}_j \neq \emptyset$)}
      \State \textbf{return} \emph{\textsc{Null}} 
    \EndIf
  \EndFor
\EndFor
\State $\mathcal{C} \leftarrow \emptyset$ \textit{// Initialize the set of constraints}
\For{$j \in \mathcal{A} \setminus M$}
  \State $\mathcal{C} \leftarrow \mathcal{C} \cup \left\{p_j \leq 0 \right\}$ \emph{// Price the unsold items at zero}
\EndFor
\For{$i \in \{1, \ldots, n\}$} 
  \State $\mathcal{C} \leftarrow \mathcal{C} \cup \left\{\sum_{j \in \vec{x}_i} p_j \leq 1, - \sum_{j \in \vec{x}_i} p_j \leq -1\right\}$
  \If{($D_i \not \subseteq \vec{x}_i$)}
    \State $\mathcal{C} \leftarrow \mathcal{C} \cup \left\{ - \sum_{j \in \vec{x}_i} p_j \leq -1 - \epsilon\right\}$ \emph{// If buyer $i$ does not get its demand set, then it's because the bundle is too expensive}
  \EndIf
\EndFor
\State \textbf{return} $\mbox{\textsc{Solve}}(\max \epsilon, \mathcal{C}, \vec{p} \geq 0)$ \emph{// Linear program solver}
\end{algorithmic}
\end{algorithm}

Finally, we investigate the problem of computing both market equilibrium allocation and prices given an instance. We will later also discuss improving the efficiency of the computed equilibria.

%\vspace{-2mm}
\begin{algorithm}[h!]
\begin{algorithmic}
\caption{\textsc{Compute-Equilibrium($\mathcal{M}$)}}
\label{alg2}
\State \textbf{input}: Market $\mathcal{M}$ with Leontief valuations
\State \textbf{output}: Equilibrium allocation and prices $(\vec{x},\vec{p})$, or \emph{\textsc{Null}} if none exist
\If{($m < n)$} 
  \State \textbf{return} \emph{\textsc{Null}} \emph{// No equilibrium : too few items}
\EndIf
\For{$\left(\;i,j \in N\right)$}
  \If{$\left(i \neq j\right.$ \textbf{\emph{and}} $D_i = D_j$ \textbf{\emph{and}} $\left.|D_i| = 1\right)$}
    \State \textbf{return} \emph{\textsc{Null}} \emph{// No equilibrium : buyers $i$ and $j$ have identical singleton demand sets}
  \EndIf
\EndFor
\State $\mathcal{A} \leftarrow \emptyset$ \emph{ // Items allocated so far}
%For each buyer i in increasing order by the number of required items that haven't been allocated so far:
\For{$\left(\mbox{buyer} \; i \in N\right.$ \emph{in increasing order by} $\left.|D_i|\right)$}
    \If{$\left(|D_i \setminus \mathcal{A}| \geq 1\right)$}
      \State $k_i \leftarrow \argmin_{\ell \in D_i \setminus \mathcal{A}} \ell$ \emph{// If not all the items in buyer $i$'s demand set have been allocated, give the buyer one of them}
    \Else
      \State $k_i \leftarrow \argmin_{\ell \in M \setminus \mathcal{A}} \ell$ \emph{// Otherwise, $i$ gets an arbitrary unallocated item}
    \EndIf
    \State $\vec{x}_i \leftarrow \{k_i\}$
    \State $\mathcal{A} \leftarrow \mathcal{A} \cup \{k_i\}$ 
\EndFor
\State $L \leftarrow \argmax_{i \in N} |D_i|$ \emph{ // The buyer with the largest demand also gets all the unallocated items (if any)}
\State $\vec{x}_L \leftarrow \vec{x}_L \cup \left( M \setminus \mathcal{A} \right)$
\For{$\left(i \in N\right)$}
  \For{$\left(j \in \vec{x}_i\right)$}
    \State $p_j \leftarrow 1 / |\vec{x}_i|$
  \EndFor
\EndFor
\State \textbf{return} $(\vec{x}, \vec{p})$
\end{algorithmic}
\end{algorithm}

\begin{theorem} \label{thm:Leontief_I}
	Given a market $\mathcal{M} = (N, M, \vec{v})$ with Leontief utilities and indivisible items, a competitive equilibrium from equal incomes exists if and only if the following hold:
	\begin{itemize}
		\item There are at least as many items as buyers ($m \geq n$)
		\item No two buyers have identical demand sets of size one.
	\end{itemize}
	Moreover, an equilibrium can be computed in polynomial time if it exists.
\end{theorem}
\begin{proof}
	Clearly the two conditions are necessary; if there are fewer items than buyers, then the budgets can never be exhausted, while if there exist two buyers whose demand sets are identical and consist
	of exactly the same item, at least one of them will be envious under any pair of feasible allocation and prices.
	
	To see that the conditions are also sufficient, consider the allocation produced by Algorithm~\ref{alg2}.
	At a high level, the algorithm first sorts the buyers in increasing order by the sizes of their demand sets, breaking ties lexicographically. Then each buyer $i$ in this order is given one item, $j$, selected from the unallocated items in the buyer's demand set (if possible), and an arbitrary un-allocated item otherwise. Finally, the last buyer (i.e. with
	the largest demand set) additionally gets all the items that remained unallocated at the end of this iteration (if any).
	
	The prices are set as follows. For each buyer $i$ among the first $n-1$ allocated, the items in its bundle, $\vec{x}_i$, are priced equally, at $1/|\vec{x}_i|$. For the last allocated buyer, $L$, the items in $\vec{x}_L \cap D_L$ are priced high (at $(1-\epsilon) / |\vec{x}_L \cap D_L|$), while the unwanted items, in $\vec{x}_L \setminus D_L$, are priced low (at $\epsilon / |\vec{x}_L \setminus D_L|$).
	
	Now we verify that the allocation and prices $(\vec{x}, \vec{p})$ computed by Algorithm~\ref{alg2} represent indeed a market equilibrium for $\mathcal{M}$ (if one exists):\\
	
	$\bullet$\hspace{2mm}$\textsc{Budgets exhausted}$: Each buyer gets a non-empty bundle priced at $1$.\\
	%$$\vec{p}(\vec{x}_i) = \left(\frac{1}{|\vec{x}_i|}\right) \cdot |\vec{x}_i| = 1.$$ 
	%Thus the buyer spends all its money.
	
	$\bullet$\hspace{2mm}$\textsc{Items sold}$: Each item is allocated by the algorithm.\\
	
	$\bullet$\hspace{2mm}$\textsc{Optimality for each buyer}$: We show that each buyer $i$ either gets its demand set or cannot afford it using a few cases:\\
	
	\hspace{3mm}$\circ$\hspace{2mm}\emph{Case 1}: 
	$(|D_i| = 1)$. Since there are no two identical demand sets with the size of one, buyer $i$ gets the unique item in its demand set, and this allocation maximizes $i$'s utility. \\
	
	\hspace{3mm}$\circ$\hspace{2mm} \emph{Case 2}: $(|D_i| \geq 2)$ and $i$ is not the last buyer.
	Then if $i$ gets an item from its demand set, 
	since $|D_i| \geq 2$ and all items are positively priced, the bundle $D_i$ is too expensive: $\vec{p}(D_i) > 1$. Otherwise, $i$ gets an item outside of its demand set.
	Then all the items in $D_i$ were allocated to the previous buyers. Since $|D_i| \geq 2$ and each previously allocated item has price $1$, $D_i$ is too expensive: $\vec{p}(D_i) > 1$. \\
	
	\hspace{3mm}$\circ$\hspace{2mm}\emph{Case 3}: $(|D_i| \geq 2)$ and $i$ is the last buyer. If $i$ does not get all its demand, then some item in $D_i$ was given to an earlier buyer at price $1$. From $|D_i| \geq 2$, 
	there is at least one other desired item in $D_i$ positively priced,
	thus $\vec{p}(D_i) > 1$. \\
	
	Thus, Algorithm~\ref{alg2} computes an equilibrium, which completes the proof.
\end{proof}

To gain more intuition, we illustrate Algorithm~\ref{alg2} on an example.

\begin{example}
	Consider a market with buyers
	$N = \{1, \ldots, 6\}$, items $M = \{1, \ldots, 8\}$, and demands:
	$D_1 = \{1\}$, $D_2 = \{2\}$, $D_3 = \{2,3\}$, $D_4 = \{2,3\}$, $D_5 = \{4, 5, 6\}$, $D_6 = \{6, 7, 8\}$.
	Algorithm~\ref{alg2} sorts the buyers in increasing order of the sizes of their demand sets, breaking ties lexicographically. The order is: $(1, 2, 3, 4, 5, 6)$.
	\begin{itemize}
		\item{} \textit{Step 1} : 
		Buyer $1$ gets item $1$ at price $1$: $\vec{x}_1 = \{1\}$, $p_1 = 1$. 
		\item{} \textit{Step 2} : Buyer $2$ gets item $2$ at price $1$ : $\vec{x}_2 = \{2\}$, $p_2 = 1$.
		\item{} \textit{Step 3} : There is one unallocated item left from buyer $3$'s demand set, and so $3$ gets it: $\vec{x}_3 = \{3\}$ and $p_3 = 1$.
		\item{} \textit{Step 4} : Buyer $4$'s demand set has been completely allocated, thus $4$ gets the free item (outside of its demand) with smallest index: $\vec{x}_4 = \{4\}$ and $p_4 = 1$. 
		\item{} \textit{Step 5} : There are two items ($5$ and $6$) left unallocated in buyer $5$'s demand set. Thus : $\vec{x}_5 = \{5\}$ and $p_5 = 1$.
		\item{} \textit{Step 6} : Buyer $6$ gets the leftover: $\vec{x}_6 = \{6, 7, 8\}$ at $p_6 = p_7 = p_8 = 1/3$. 
	\end{itemize}
\end{example}

The characterization obtained through Algorithm~\ref{alg2} raises several important questions. For example, not only do 
fair division procedures typically guarantee fairness (according to a given solution concept), but also they improve some measure of efficiency when possible. \\

The utilitarian \emph{social welfare} of an allocation $\vec{x}$ is defined as the sum of the buyers' utilities:
$
\textsc{SW}(\vec{x}) = \sum_{i =1}^{n} u_i(\vec{x}_i)
$.
Note that social welfare normalization is not required for any of our next results. \\

%For measuring social welfare, valuations must be normalized such that players are weighted equally (since their rights over the goods are equal), and so: %$u_i(M) = 1$, for each buyer $i$. This can also be interpreted as the number of buyers that receive their demand sets (possibly in addition to other items).
%Then the goal is to find allocations that are not only fair, but also maximize (or approximate) the optimal social welfare among all the fair allocations. 
As the following example illustrates, the allocation computed by Algorithm~\ref{alg2} can be the worst possible
among all the market equilibria.

\begin{example}
	Given $n \in \mathbb{N}$, let $N = \{1, \ldots, n\}$ be the set of buyers, $M = \{1, \ldots, 2n\}$ the set of items, and the demand sets given by: $D_i = \{2i-1, 2i\}$, for each $i \in N$. Algorithm~\ref{alg2} computes the allocation:
	$\vec{x}_1 = \{1\}$, $\vec{x}_2 = \{3\}$, \ldots, $\vec{x}_{n-1} = \{2n-3\}$, $\vec{x}_n = \{2, 4, \ldots, 2n-2, 2n-1, 2n\}$, with a social welfare of $\textsc{SW}(\vec{x}) = 1$.
	The optimal allocation supported in a competitive equilibrium is:
	$\vec{x}^*_i = \{2i-1, 2i\}$, for each $i \in N$, with a social welfare of $\textsc{SW}(\vec{x}^*) = n$.
\end{example}

These observations give rise to the question:
\emph{Is there an efficient algorithm for computing a competitive equilibrium from equal incomes with optimal social welfare (among all equilibria) for perfect complements with indivisible goods?} \\

Note that the allocation that maximizes social welfare among all possible allocations cannot always be supported in a
competitive equilibrium. We illustrate this phenomenon in Example~\ref{eg:maxWelfare}.

\begin{example} \label{eg:maxWelfare}
	Consider a market with buyers: $N = \{1,2\}$ and items: $M = \{1,2,3\}$, where the demand sets are: $D_1 = D_2 = \{1,2\}$. 
	Concretely, let these demands be induced by the valuations: $v_{1,1} = v_{1,2} = 1$, $v_{1,3} = 0$ and $v_{2,1} = v_{2,2} = 1$, $v_{2,3} = 0$.
	The optimal social welfare is $1$ and can be achieved by giving one of the buyers its entire demand set and the other buyer the remaining item; for example, let $\vec{x}_1^* = \{1,2\}$ 
	and $\vec{x}_2^* = \{3\}$, with $p_1 = p_2 = 1/2$ and $p_3 = 1$. Clearly no such allocation can be supported in an equilibrium, because whenever a buyer gets their full demand, the other buyer does not get its own demand but can afford it (their initial budgets are equal). Thus every competitive equilibrium has a social welfare of zero, such as $\vec{x}_1 = \{1\}$, $\vec{x}_2 = \{2,3\}$, with $p_1 = 1$, $p_2 = 1$, $p_3 = 0$.
\end{example}

The next result implies that equilibria with optimal social welfare cannot be computed efficiently in the worst case.

\begin{theorem} \label{thm:optWelfare}
	Given a market $\mathcal{M} = (N, M, \vec{v})$ with Leontief valuations, indivisible goods,  and an integer $K \in \mathbb{N}$, it is NP-complete to decide if $\mathcal{M}$ has a competitive equilibrium 
	from equal incomes with social welfare at least $K$.
\end{theorem}
\begin{proof}(sketch)
	We use a reduction from the NP-complete problem $\textsc{SET}$ $\textsc{PACKING}$: 
	\begin{quote}
		\emph{Given a collection $\mathcal{C} = \langle C_1, \ldots, C_n\rangle$ of finite sets and a positive integer $K \leq n$, does $\mathcal{C}$ 
			contain at least $K$ mutually disjoint sets?
		}
	\end{quote}
	Given collection $\mathcal{C}$ and integer $K$, let $\mathcal{M}$ be a market with buyers $N = \{1, \ldots, n\}$, items $M = \{1, \ldots, m+n\}$, and demands $D_i = C_i \cup \{m+i\}$, for all $i \in N$.
	It can be checked that $\mathcal{M}$ has a competitive equilibrium with social welfare at least $K$ if and only if $\mathcal{C}$ has a disjoint collection of at least $K$ sets.
\end{proof}
We obtain a $1/n$-approximation of the optimal welfare in polynomial time (Algorithm~\ref{alg2.1}) and leave open the question of determining the tight bound; see appendix for algorithm and its proof.

\begin{theorem} \label{thm:additive_maxSW_apx}
	There is a polynomial-time algorithm that computes a competitive equilibrium from equal incomes with a social welfare of at least $1/n$ of the optimum welfare attainable in any equilibrium.
\end{theorem}

\section{Perfect Substitutes}

We begin by introducing the utility function in a market with perfect substitutes, represented through additive valuations.
\begin{definition}[Additive Utility for Indivisible Goods]
Given a market $\mathcal{M} = (N, M, \vec{v})$ with additive utilities and indivisible goods, the utility of a buyer $i$ for a bundle $\vec{x} = \langle x_1, \ldots, x_m \rangle 
\in \{0,1\}^m$ is:
\begin{equation} \label{eq:AdditiveUtility}
u_i(\vec{x}) = \sum_{j=1}^{m} v_{i,j} \cdot x_{i,j}
\end{equation}
where $v_{i,j}$ are constants and represent the value of buyer $i$ for consuming one unit of good $j$, while $x_{i,j} = 1$ if buyer $i$ gets good $j$, and $x_{i,j}=0$, otherwise.
\end{definition}

Next we investigate the computation of competitive equilibria from equal incomes with indivisible goods and additive utilities. Note that if a market $\mathcal{M}$ has a competitive equilibrium at some allocation and prices $(\vec{x}, \vec{p})$, then $\mathcal{M}$ is guaranteed to have an equilibrium at the same allocation $\vec{x}$ where all the prices are rational numbers, $(\vec{x}, \vec{p}^{*})$; this aspect appears implicitly in some of the following proofs.

\begin{theorem}
Given a market $\mathcal{M} = (N, M, \vec{v})$ with additive valuations, indivisible goods, and tuple $(\vec{x}, \vec{p})$, where $\vec{x}$
is an allocation and $\vec{p}$ is a price vector, it is coNP-complete to determine whether $(\vec{x}, \vec{p})$ is a competitive equilibrium for $\mathcal{M}$.
\end{theorem}
\begin{proof}
The problem admits efficiently verifiable ``no'' instances: it can be checked in polynomial time if the allocation is not feasible, or the budgets are not exhausted, or not all the items are 
sold. Otherwise, 
if $(\vec{x}, \vec{p})$ is not a market equilibrium for $\mathcal{M}$, then 
there exists a buyer $k$ with a suboptimal bundle. In other words, the certificate that $(\vec{x}, \vec{p})$ is not a market equilibrium is given by 
a tuple $(k, D)$, where $k$ is a buyer that strictly prefers bundle $D \subseteq M$ to $\vec{x}_k$ and can also afford it; that is,  $u_k(\vec{x}_k) < u_k(D)$
and $p(D) \leq 1$. \\

We show hardness using the $\textsc{SUBSET-SUM}$ problem:
\begin{quote}
\emph{Given a set of positive integers $\mathcal{W} = \{w_1, \ldots, w_n\}$ and a target number $K$, is there a subset $S \subseteq \mathcal{W}$ that adds up to exactly $K$?}
\end{quote}

Given $\langle \mathcal{W}, K\rangle$, construct market $\mathcal{M} = (N,M, \vec{v})$ and tuple $(\vec{x}, \vec{p})$, with buyers $N = \{0, 1, \ldots, n\}$, items $M = \{0, 1, \ldots, 2n\}$, and values:
\begin{itemize}
\item Buyer $0$: $v_{0,0} = K - 1$; $v_{0,j} = w_j$, for all $j \in \{1, \ldots, n\}$; $v_{0,j} = 0$, for all $j \in \{n+1, \ldots, 2n\}$.
\item Buyer $i \in \{1, \ldots, n\}$: $v_{i,n+i} = 1$; $v_{i,j} = 0$, for all $j \in M \setminus \{n+i\}$.
\end{itemize}

Let $\vec{x}_0 = \{0\}$ and $\vec{x}_i = \{i, n+i\}$, for all $i \in \{1, \ldots, n\}$. Define prices: $p_0 = 1$, $p_j = \frac{w_j}{K}$ and $p_{n+j} = 1 - p_j$, for all $j \in \{1, \ldots, n\}$ 
(Note that if there exist items with $w_j > K$, those items can be thrown away from the beginning).

($\implies$) If there is a solution $S \subseteq U$ to $\mathcal{W}$, then we claim that $\mathcal{M}$ does not have an equilibrium at $(\vec{x}, \vec{p})$ since  
buyer $0$ can acquire a better bundle, namely $S$:

$\bullet$\hspace{2mm}Buyer $0$ can afford $S$: $$\sum_{j \in S} p_j = \sum_{j \in S} \frac{w_j}{K} = \frac{K}{K} = 1.$$

$\bullet$\hspace{2mm}Bundle $S$ is strictly better than $\vec{x}_0$: $$\sum_{j \in S} v_{0,j} = \sum_{j \in S} w_j = K > K-1 = v_{0,0}.$$

($\impliedby$) If $(\vec{x}, p)$ is not a market equilibrium, then it must be that buyer $0$ can get a better bundle (since all budgets are spent, all items are sold, and the other buyers already have their unique valuable item).

Thus there is bundle $S$ such that: $(i)$ $\sum_{j \in S} v_{0,j} > v_{0,0}$ and $(ii)$ $\sum_{j \in S} p_j \leq 1$. From $v_{0,j} = 0$, for all $j \in \{n+1, \ldots, 2n\}$, it follows that
$S \subset \{1, \ldots, n\}$ (otherwise, just take $S' = S \cap \{1, \ldots, n\}$).
Condition $(i)$ is equivalent to:
$\sum_{j \in S} w_j \geq K$
and condition $(ii)$ can be rewritten as:
$$\sum_{j \in S} \frac{w_j}{K} \leq 1 \iff \sum_{j \in S} w_j \leq K$$
Then $S$ is a subset-sum solution; this completes the proof.
\end{proof}

\begin{theorem} \label{thm:additive_I}
Given a market with indivisible goods and additive valuations, $\mathcal{M} = (N, M, \vec{v})$, it is NP-hard to decide if $\mathcal{M}$ has a competitive equilibrium.
\end{theorem}
\begin{proof}
We reduce from the NP-complete problem EXACT COVER BY $3$-SETS (X3C):
\begin{quote}
\emph{
Given universe $\mathcal{U} = \{1, \ldots, 3n\}$ of elements and family of subsets $\mathcal{F} = \{S_1, \ldots, S_k\}$, with $|S_i| = 3$, $\forall i$,
decide if there is collection $S \subseteq \mathcal{F}$ such that each element of $\mathcal{U}$ occurs exactly once in $S$.}
\end{quote}
Given X3C instance $\langle \mathcal{U}, \mathcal{F}\rangle$, define $N = \{1, \ldots, k\}$, $M = \{1, \ldots, 3n, 3n+1, \ldots, 2n+k\}$
(note this assumes that $k \leq n$, since otherwise
the answer to the X3C instance is trivially ``no''), and valuations for each buyer $i\in N$:
\begin{itemize}
\item $v_{i,j} = \frac{1}{3}$, for all $j \in S_i$.
\item $v_{i,j} = 1$, for all $j \in \{3n+1, \ldots, 2n+k\}$.
\item $v_{i,j} = 0$, otherwise.
\end{itemize}

If the market has some competitive equilibrium $(\vec{x}, \vec{p})$, then the following conditions hold:
\begin{itemize}
\item Each buyer gets a bundle worth at least $1$, since the items in $\{3n+1, \ldots, 2n+k\}$ are each worth $1$ to every buyer and each of their prices is at most $1$ (since all items get sold).
\item No buyer can get a bundle worth more than $1$. % \simina{TODO: fill in.}
\end{itemize}
Then each buyer gets a bundle worth exactly $1$, and so the items in $\{3n+1, \ldots, 2n+k\}$ are priced at $1$ each. The remaining $n$ buyers get a bundle worth $1$ each from the items $\{1, \ldots, 3n\}$, 
which can only happen if their allocations form a solution to the X3C instance. % (\simina{A bit more detail}).

If X3C has a solution $S$, then a market equilibrium is obtained immediately by giving the sets in $S$ to the buyers that want them, and the leftover items, in $\{3n+1, \ldots, 2n+k\}$, to the remaining buyers. % (\simina{A bit more detail}).
\end{proof}

The next question, of computing an equilibrium allocation given a market $\mathcal{M}$ and a price vector $\vec{p}$ was raised by Bouveret and Lema\^{i}tre (\cite{BL14}). In a recent note, Aziz (\cite{Aziz15}) 
also studied the hardness of this problem. We include our proof as well, which uses the 
$\textsc{PARTITION}$ problem.

\begin{theorem} \label{thm:additive_Ip}
Given a market $\mathcal{M} = (N, M, \vec{v})$ with indivisible goods, additive valuations, and price vector $\vec{p}$, it is coNP-hard to decide if there is an allocation $\vec{x}$ such that $(\vec{x}, \vec{p})$ is a market equilibrium.
\end{theorem}
\begin{proof}
We use a reduction from the NP-complete problem $\textsc{PARTITION}$.
Given a set $S = \{s_1, \ldots, s_m\}$, where $\sum_{j=1}^{m} s_j = 2V$ and $s_j \in \mathbb{N}$, $\forall j \in \{1, \ldots, m\}$, construct the following market with indivisible goods and additive valuations:
\begin{itemize}
\item Let $N = \{1, 2\}$ and $M = \{1, \ldots, m+2\}$.
\item Buyer $1$'s valuations: $v_{1,j} = s_j$, $\forall j \in \{1, \ldots, m\}$, $v_{1,m+1} = 3V$ and $v_{1,m+2} = V - 1$.
\item Buyer $2$'s valuations: $v_{2, j} = 1$, $\forall j \in \{1, \ldots, m\}$ and $v_{2,m+1} = v_{2,m+2} = 0$.
\end{itemize}
Consider the price vector given by $p_j = \frac{s_j}{2V}$, $\forall j \in \{1, \ldots, m\}$, $p_{m+1} = \frac{1}{2} = p_{m+2}$.

We claim that $S$ has a partition if and only if $\mathcal{M}$ does not have an equilibrium at $\vec{p}$. 
First, note that buyer $2$ can afford to buy all the items it has a strictly positive value for -- i.e. the set $M' = \{1, \ldots, m\}$ -- since:
\[
\vec{p}(M') = \sum_{j=1}^{m} p_j = \sum_{j=1}^{m} \frac{s_j}{2V} = \frac{2V}{2V} = 1
\]
Thus any equilibrium allocation $\vec{x}$ has the property that $M' \subseteq \vec{x}_2$. In addition, buyer $2$ cannot afford any other item, and so it must be the case that $\vec{x}_1 = M \setminus M' = \{m+1, m+2\}$ and $\vec{x}_2 = M'$ in any equilibrium.

($\implies$) If there is a partition $\langle A, B \rangle$ of $S$, then buyer $1$ can afford a better bundle at these prices, namely $Y = \{m+1\} \cup A$, since:
\begin{eqnarray*}
v_{1}(Y) &=& v_{1, m+1} + \sum_{j \in A} v_{1,j} = 3V + \sum_{j \in A} s_j \\
&= & 4V > 4V-1 = v_{1}(\vec{x}_1)
\end{eqnarray*}
and 
\[
\vec{p}(Y) = p_{m+1} + \sum_{j \in A} p_j = \frac{1}{2} + \sum_{j \in A} \frac{s_j}{2V} = \frac{1}{2} + \frac{V}{2V} = 1
\]
Thus the market cannot have a competitive equilibrium at $\vec{p}$.

($\impliedby$) If the market does not have an equilibrium at $\vec{p}$, then it must be the case that in any feasible allocation there exists an improving deviation. Consider the allocation $\vec{x}_1 = \{m+1, m+2\}$ and $\vec{x}_2 = \{1, \ldots, m\}$. Since buyer $1$ is already getting its optimal bundle, it follows that buyer $1$ has an improving deviation. We consider a few cases:
\begin{itemize}
\item If buyer $1$ replaces both items $m+1$ and $m+2$, then the only bundle it can afford by doing this is $\vec{x}_2$ and $v_1(\vec{x}_2) < v_1(\vec{x}_1)$; thus buyer $1$ does not have an improving deviation of this type.
\item If buyer $1$ replaces item $m+1$ with some subset $C$ of $M'$ then again its utility decreases since:
\begin{eqnarray*}
v_1\left(\{m+2\} \cup C\right) & = & V-1 + v_1(C) \\
& =& V - 1 + \sum_{j \in C} s_j \\
&<& V - 1 + 2V  < 4V -1 \\
&=& v_1(\vec{x}_1)
\end{eqnarray*}
\item The only type of deviation left is the one where buyer $1$ replaces item $m+2$ with some subset $C$ of $M'$. Then the only improvements in value can come from bundles
worth at least $V$. That is, there must exist a subset $C \subset M'$ with the property that:
\begin{eqnarray*}
v_1(C) &=& \sum_{j \in C} s_j > v_1\left(\{m+2\}\right) = V-1 \\
&\iff & \sum_{j \in C} s_j \geq V
\end{eqnarray*}
and 
\[
\vec{p}(C) = \sum_{j \in C} \frac{s_j}{2V} \leq \frac{1}{2} \iff \sum_{j \in C} s_j \leq V
\]
It follows that $\sum_{j \in C} s_j = V$ and so $\langle C, M' \setminus C \rangle$ are a partition of $S$.
\end{itemize}
This completes the proof of the theorem.
\end{proof}

Our final proof is the most subtle and included next.

\begin{theorem}
Given a market $\mathcal{M} = (N, M, \vec{v})$ with indivisible goods, additive valuations, and allocation $\vec{x}$, it is coNP-hard to decide if there is a price vector $\vec{p}$ such that $(\vec{x}, \vec{p})$ is a market equilibrium.
\end{theorem}
\begin{proof}
We use the NP-complete problem $\textsc{SUBSET-SUM}$. Given set of positive integers $\mathcal{W} = \{w_1, \ldots, w_m\}$ and integer $K$,
%input $\mathcal{W} = \langle U, \vec{v}, \vec{w}, V, W \rangle$, where $U = \{1, \ldots, m\}$ is a set of items of values $\vec{v}$ and weights $\vec{w}$, $V$ a value goal, and $W$ a maximum weight constraint,
we construct a market $\mathcal{M} = (N, M, \vec{v})$ and an allocation $\vec{x}$, such that an equilibrium price vector exists at $\vec{x}$ if and only if the subset-sum problem does not have a solution. 

Let $N = \{1, 2\}$, $M = \{1, \ldots, m+2\}$, allocation $\vec{x}$ given by $\vec{x}_1 = \{m+1, m+2\}$, $\vec{x}_2 = \{1, \ldots, m\}$, and values:
\begin{itemize}
\item Buyer 1: $v_{1,m+1} = K-1$; $v_{1,m+2} = 4 \left(\sum_{j=1}^{m} w_j\right)^2$; $v_{1,j} = w_j$, for all $j \in \{1, \ldots, m\}$.
\item Buyer 2: $v_{2,m+1} = K + 1$; $v_{2,m+2} = 0$; $v_{2,j} = w_j$, for all $j \in \{1, \ldots, m\}$.
\end{itemize}

Note we can assume the sum of the numbers in $\mathcal{W}$ is at least $K$ and none is greater than $K$.

$(\implies)$ If there is a solution $S$ to $\langle \mathcal{W}, K\rangle$, then we claim there can be no market equilibrium. Let $\vec{p}$ be any feasible price vector. Then
the utility of buyer $1$ for bundle $S$ is: $u_1(S) = \sum_{j \in S} v_{1,j} = \sum_{j \in S} w_j = K > K-1 = v_{1,m+1}$.

By the equilibrium property, it must be the case that buyer $1$ cannot afford to swap pay for item $m+1$ instead of the set $S$, and so: $\vec{p}(S) > \vec{p}_{m+1}$. 
However, this implies buyer $2$ can afford to swap the set $S$ with item $m+1$, and moreover, this is an improving deviation since: 
%\[
$v_{2,m+1} = K + 1 > K = \sum_{j \in S} w_j$.
%\]
Thus there can be no equilibrium prices.

$(\impliedby)$ If there is no market equilibrium, then we claim there is a subset-sum solution.
To this end, we show that whenever there is no set $S\subseteq U$ such that $\sum_{j \in S} w_j = K$, then a market equilibrium exists. For example, define the next price vector (at which all the budgets are spent):
\begin{itemize}
\item $p_j =  \frac{w_j}{\sum_{k=1}^{m} w_k}$, for all $j \in \{1, \ldots, m\}$
\item $p_{m+1} = \frac{K - 1 + \epsilon}{\sum_{k=1}^{m} w_k}$, where $\epsilon = \frac{1}{4(m+1)^2}$
\item $p_{m+2} = 1 - p_{m+1}$
\end{itemize}
First we claim that buyer $1$ does not have a deviation. Note that item $m+2$ is very valuable, i.e. buyer $2$ would never exchange it for any subset of $\{1, \ldots, m\}$. Thus the only remaining type of deviation is one in which buyer $1$ exchanges item $m+1$ for a subset $S \subseteq \{1, \ldots, m\}$. Then it hold that $u_1(S) > v_{1,m+1} = K-1$, that is, $u_1(S) \geq K$. We have:
\begin{eqnarray*} \label{eq:nodev1}
& & \vec{p}(S) - p_{m+1} = \left( \sum_{j \in S} p_j \right)  - p_{m+1} \\
&=& \sum_{j \in S} \left( \frac{w_j}{\sum_{k=1}^{m} w_k} \right) - \frac{K-1+\epsilon}{\sum_{k=1}^{m} w_k} > 0\\
& \iff &\sum_{j \in S} w_j > K - 1 + \epsilon
\end{eqnarray*}

The last inequality holds since $u_1(S) \geq K \iff \sum_{j \in S} w_j \geq K > K - 1 + \epsilon$. Thus bundle $S$ is too expensive for buyer $1$ to afford it with the price of item $m+1$.

Second, the only type of improving deviation of buyer $2$ is one in which a set $S \subseteq \{1, \ldots, m\}$ is exchanged for item $m+1$. For this to be an improvement, it must hold that:
\begin{eqnarray*}
v_{2,m+1} > u_2(S) \iff u_2(S) \leq K \iff \sum_{j \in S} w_j \leq K 
\end{eqnarray*}
Since there is no subset-sum solution, we have: $\sum_{j \in S} w_j < K$, and so: $\sum_{j \in S} w_j \leq K - 1$.
Equivalently:
\begin{displaymath}
\vec{p}(S) = \sum_{j \in S} \left( \frac{w_j}{\sum_{k=1}^{m} w_k} \right) \leq \frac{K - 1}{\sum_{k=1}^{m} w_k}
< \frac{K - 1 + \epsilon}{\sum_{k=1}^{m}w_k} = p_{m+1}
\end{displaymath}

It follows that $p_{m+1} > \vec{p}(S)$, and so buyer $2$ cannot afford to exchange the set $S$ for item $m+1$.
Thus neither buyer $1$ nor buyer $2$ have a deviation, and so $(\vec{x}, \vec{p})$ is an equilibrium.
Then there is a subset-sum solution if and only if the market does not have an equilibrium, which completes the proof.
\end{proof}

\section{Acknowledgements}
Simina Br\^anzei and Peter Bro Miltersen 
acknowledge support 
from the Danish National Research Foundation
and The National Science Foundation of China (under the grant 61361136003) for
the Sino-Danish Center for the Theory of Interactive Computation and from the Center for
Research in Foundations of Electronic Markets (CFEM), supported by the Danish
Strategic Research Council.

\bibliographystyle{plain}

%\bibliography{abb,ultimate}

\section{Appendix}

In this section we include the algorithm that computes a CEEI with a $1/n$ approximation for the optimal social welfare (among all equilibria).

%\vspace{-4mm}
\begin{algorithm}[h!] 
	\begin{algorithmic}
		\caption{\textsc{Compute-Equilibrium-1/n-SW($\mathcal{M}$)}}
		\label{alg2.1}
		\State \textbf{input}: Market $\mathcal{M}$ with Leontief valuations
		\State \textbf{output}: Equilibrium allocation and prices $(\vec{x},\vec{p})$, or \emph{\textsc{Null}} if none exist.
		%\emph{Guarantee: The social welfare at $(\vec{x},\vec{p})$ is at least $1/n$ of the optimal equilibrium.}\\
		
		\If{($(m < n)$ or $(\exists \; i,j \in N$ with $D_i = D_j$ and  $|D_i| = 1)$}
		\State \textbf{return} \emph{\textsc{Null}} \emph{// Too few items or two buyers have identical singleton demands}
		\EndIf
		%\If{($\exists \; i \in N$ with $|D_i| = 1$)}
		%\State \textbf{return} \emph{Compute-Equilibrium($\mathcal{M}$)} \emph{ // If there is a buyer with singleton demand, the old algorithm will do}
		%\EndIf
		\State $\mathcal{P} \leftarrow \emptyset$ \emph{ // Buyers eligible for getting their full demand}
		\For{($i \in N$)} 
		\If{($\not \exists \; k \in N \setminus \{i\}$ such that ($D_k \subseteq D_i$ \textbf{and} $m - |D_i| \geq n-1$))}
		\State $\mathcal{P} \leftarrow \mathcal{P} \cup \{i\}$
		\EndIf
		\EndFor
		\If{($\mathcal{P} = \emptyset$)}
		\State \textbf{return} \emph{Compute-Equilibrium($\mathcal{M}$)} \emph{// No good equilibrium: use basic algorithm}
		\EndIf
		\State $k \leftarrow \argmax_{i \in \mathcal{P}} v_i$ \emph{// The buyer in $\mathcal{P}$ with highest valuation for its demand gets it}
		\State $\vec{x}_{k} \leftarrow D_{k}$
		\State $\mathcal{M}' \leftarrow \left(N \setminus \{k\}, M, (D_i)_{i \neq k}\right)$ \emph{// Reduced market, with buyers $N \setminus \{k\}$}
		\State $(\vec{y}, \vec{q}) \leftarrow$ \emph{Compute-Equilibrium}($\mathcal{M}', D_k$) \emph{// Call basic algorithm on market $\mathcal{M}'$ with pre-allocated items $D_k$} 
		\State $\vec{x} \leftarrow (\vec{x}_k, \vec{y})$ \emph{// Final allocation}
		\For{($j \in D_k$)} 
		\State $p_j \leftarrow 1 / |D_k|$ \emph{// Price buyer $k$'s items uniformly}
		\EndFor
		\For{($j \in M \setminus D_k$)}
		\State $p_j \leftarrow q_j$ \emph{// For the other buyers, use the prices from the reduced market}
		\EndFor
		\State \textbf{return} $(\vec{x}, \vec{p})$
	\end{algorithmic}
\end{algorithm}
%\vspace{-4mm}

\emph{\textbf{Theorem} \ref{thm:additive_maxSW_apx} (restated)
	There is a polynomial-time algorithm that computes a competitive equilibrium from equal incomes with a social welfare of at least $1/n$ of the optimum welfare attainable in any equilibrium.}
\begin{proof}
	We claim that Algorithm~\ref{alg2.1} computes an equilibrium with social welfare at least $1/n$ of the optimal, using the \emph{Compute-Equilibrium} procedure (Algorithm \ref{alg2}) as a subroutine. Note that this approximation holds for weighted valuations 
	(i.e. the value of a buyer for the whole bundle can be arbitrary). 
	Given a market $\mathcal{M}$, Algorithm~\ref{alg2.1} computes a set $\mathcal{P}$ of \emph{eligible} buyers, i.e. buyers $k$ for which the following conditions hold:
	\begin{description}
		\item[$(1)$] No buyer $i$'s demand is completely contained in the demand of buyer $k$.
		\item[$(2)$] $m - |D_k| \geq n - 1$
	\end{description}
	
	Both conditions are necessary in any equilibrium that gives buyer $k$ its demand set. If Condition $1$ is violated, there exists buyer $i \neq k$ with $D_i \subseteq D_k$; $i$ can afford its demand but does not get it, which cannot happen in an equilibrium.
	If Condition $2$ is violated, by allocating $k$ all of its demand, too many items are used up and it's no longer possible to extract all the money.
	
	Algorithm~\ref{alg2.1} gives the eligible buyer $k \in \mathcal{P}$ with a maximal valuation for its demand all the required items. The remaining buyers and items are allocated using Algorithm~\ref{alg2}. 
	We claim that the tuple $(\vec{x}, \vec{p})$ computed is an equilibrium:\\
	
	\hspace{-6mm}$(a) \; \textsc{Budgets exhausted:}$ By combining: $(i)$ $\vec{x}_k \neq \emptyset$, $(ii)$ $m - |D_k| \geq n - 1$, and $(iii)$ the fact that \emph{Compute-Equilibrium} finds an equilibrium in the reduced market, all buyers get a non-empty bundle at a price of $1$. \\ 
	
	\hspace{-6mm}$(b)\; \textsc{Items sold:}$ Clearly all the items are allocated. \\
	
	\hspace{-6mm}$(c)\; \textsc{Optimality for each buyer:}$ If the algorithm exits before the set $\mathcal{P}$ is constructed, then no equilibrium exists. Otherwise, buyer $k$ gets its demand.
	
	Let $i \neq k$ be a buyer that does not get its demand; then $|D_i| \geq 2$. If $i$ is not the last buyer allocated, then $|\vec{x}_i| = 1$ and there are two subcases: \\
	\begin{itemize}
		\item $\vec{x}_i \subset D_i$: By construction, $\vec{p}(\vec{x}_i) = 1$, so $\vec{p}(D_i) = \vec{p}(\vec{x}_i) + \vec{p}(D_i \setminus \vec{x}_i) > 1$ since all the items are priced strictly positively.
		\item $\vec{x}_i \cap D_i = \emptyset$: Then by the time buyer $i$ was allocated, the items in $D_i$ have been exhausted. Since buyer $k$ gets its full demand before everyone else, it cannot be that $D_i \subseteq D_k$ (this holds by choice of the set $\mathcal{P}$, of eligible buyers). Thus there is at least one item from $D_i$ given to a buyer other than $k$ at a price of $1$, which combined with: $|D_i| \geq 2$ gives: $\vec{p}(D_i) > 1$.
	\end{itemize}
	Otherwise, $i$ is the last buyer allocated; note that $x_i$ is always computed in the procedure \emph{Compute-Equilibrium}. Again there are two cases:
	\begin{itemize}
		\item $\vec{x}_i \cap D_i = \emptyset$: Then $D_i$ was exhausted before computing $\vec{x}_i$. Since it cannot be that $D_i \subseteq D_k$, there exists item $j \in D_i \setminus D_k$ given to another buyer at a price of $1$. Using again the fact that $|D_i| \geq 2$ and that all prices are positive, we get: $\vec{p}(D_i) > 1$, thus $i$ cannot afford its demand.
		\item $\vec{x}_i \cap D_i \neq \emptyset$: Let $\ell = |\vec{x}_i \cap D_i|$. Then $i$ pays $1 - \epsilon$ for $\ell$ items from $D_i$, and $\epsilon$ for some other items in $\vec{x}_i \setminus D_i$. We show that $i$ cannot use $\epsilon$ money to purchase the items missing from its allocation. 
		Let $j \in D_i \setminus \vec{x}_i$ be such an item and $S' = \left( \vec{x}_i \cap D_i \right) \cup \{j\}$. Then item $j$ was allocated to a previous buyer, priced at least $1 / |D_k|$. Since $\epsilon < 1 / |D_k|$, $i$ would have to pay at least $1 - \epsilon + \frac{1}{|D_k|} > 1$ for $S'$. From $S' \subseteq D_i$, we get: $\vec{p}(D_i) \geq \vec{p}(S') >1$, and so $i$ cannot afford $D_i$.
	\end{itemize}
	
	Thus the market equilibrium conditions are met. Clearly the best equilibrium cannot have a social welfare higher than $n \cdot v_k$, which gives a $1/n$-approximation.
\end{proof}

\end{document}